\documentclass[11pt]{llncs}
\usepackage{color}
\usepackage{graphicx}
\usepackage{fullpage}
\usepackage{amsfonts}
\usepackage{latexsym,amssymb,amsmath}
\usepackage[ruled]{algorithm}
\usepackage{algorithm}
\usepackage{comment}

\def\d{~\textrm{d}}
\def\Li{\textrm{Li}}

\begin{document}

\title{Excuse Me! \\or The Courteous Theatregoers' Problem
\thanks{An extended abstract of this paper appeared in
the Proceedings of Seventh International Conference on
Fun with Algorithms,
July 1--3, 2014, Lipari Island, Sicily, Italy, Springer LNCS.}
}
\author{
Konstantinos Georgiou\inst{1}
\and
Evangelos Kranakis\inst{2}\thanks{Research supported in part by NSERC Discovery grant.}
\and
Danny Krizanc\inst{3}
}

\institute{
Department of Combinatorics \& Optimization, University of Waterloo
\and
School of Computer Science, Carleton University
\and
Department of Mathematics \& Computer Science,
Wesleyan University
}
\maketitle

\begin{abstract}
Consider a theatre consisting of $m$ rows each containing $n$ seats. 
Theatregoers enter the theatre along aisles and pick a row which they enter along one
of its two entrances  so as to occupy a seat. 
Assume they select their seats uniformly and
independently at random among the empty ones.
A row of seats is narrow and an occupant who is already
occupying a seat is blocking passage to new incoming theatregoers.
As a consequence, occupying a specific seat depends on
the courtesy of theatregoers and their
willingness to get up so as to create  free space
that will allow passage to others. Thus,    
courtesy facilitates and may well increase the overall
seat occupancy of the theatre. We say a theatregoer
is {\em courteous} if (s)he will get up to let others pass. Otherwise,
the theatregoer is {\em selfish}. 
A set of theatregoers 
is {\em courteous with probability $p$} (or {\em $p$-courteous}, for
short) if 
each theatregoer in the set is courteous with probability $p$, randomly and
independently.
It is assumed that the behaviour of a theatregoer
does not change during the occupancy of the row.
Thus, $p=1$ represents the case where all theatregoers are 
courteous and $p=0$ when they are all selfish. 

In this paper, we
are interested in the following question:
what is the expected number of
occupied seats as a function of the total
number of seats in a theatre, $n$, and the probability that a theatregoer
is courteous, $p$? 
We study and analyze interesting variants of this problem
reflecting behaviour of the theatregoers as entirely selfish, 
and  $p$-courteous for a row of seats 
with one or two entrances and as a consequence for a theatre with $m$ rows 
of seats with multiple
aisles.  We also consider the case where seats in a row are chosen according
to the geometric distribution and the Zipf distibrution (as opposed to the uniform distribution) and provide
bounds on the occupancy of a row (and thus the theatre) in each case. 
Finally, we propose
several open problems for other seating probability
distributions and theatre seating arrangements.

\vspace{0.5cm}
\end{abstract}


\section{Introduction}

A group of Greek tourists is vacationing on the island of Lipari and they find
out that the latest release of their favourite playwright is playing 
at the local theatre (see Figure~\ref{lipari-fig}), {\em Ecclesiazusae} by Aristophanes, a big winner at last year's (391 BC)
Festival of Dionysus. Seating at the theatre is open (i.e., the seats are chosen by the audience members as 
they enter). The question arises as to whether they will be able to find seats. As it turns out
this depends upon just how courteous the other theatregoers are that night. 

%

Consider a theatre with $m$ rows containing $n$ seats each. 
Theatregoers enter the theatre along aisles, choose a row, and enter it from one of its ends, wishing
to occupy a seat. They
select their seat in the row uniformly and
independently at random among the empty ones.
The rows of seats are narrow and
if an already sitting theatregoer is not willing to get up
then s(he) blocks passage to the selected seat and the incoming theatregoer
is forced to select a seat among unoccupied seats
between the row entrance and the theatregoer who refuses to budge. 
Thus, the selection and overall occupancy of seats depends on
the courtesy of sitting theatregoers, i.e.,  their
willingness to get up so as to create  free space
that will allow other theatregoers go by.

%

An impolite theatregoer, i.e., one that never gets up from
a position s(he) already occupies, is referred to as {\em selfish} theatregoer.
Polite theatregoers (those that will get up to let someone pass) are
referred to as {\em courteous}. 
On a given evening we expect some fraction of the audience to be selfish
and the remainder to be courteous. We say a set of theatregoers is
{\em $p$-courteous} 
if each individual in the set is courteous with probability $p$ and 
selfish with probability $1-p$.
We assume that the status of a
theatregoer (i.e., selfish or courteous)
is independent of the other theatregoers and it
remains the same throughout the
occupancy of the row. Furthermore, theatregoers select a vacant seat uniformly
at random. They enter a row from one end
and inquire (``Excuse me''), if necessary, whether an already
sitting theatregoer is courteous enough to let him/her go by and occupy
the seat selected. If a selfish theatregoer is encountered, a seat is selected
at random among the available unoccupied ones, should any exist. 
We are interested in the following question:
\begin{quote}
What is the expected number of seats
occupied by theatregoers when all
new seats are blocked,
as a function of the total
number of seats and the theatregoers' probability $p$ of being courteous? 
\end{quote}

We first study the problem on
a single row with either one entrance or two.
For the case $p=1$ it is easy to see that the row will be fully occupied when
the process finishes. We show that for $p=0$ (i.e., all theatregoers are selfish)
the expected number of occupied seats is only $2 \ln n + O(1)$ 
for a row with
two entrances. Surprisingly, for any fixed $p<1$ we show that this is only 
improved by essentially a constant factor of $\frac{1}{1-p}$. 

Some may
argue that the assumption of choosing seats uniformly at random is somewhat
unrealistic. People choose their seats for a number of reasons (sight lines, privacy, etc.)
which may  result in a nonuniform occupancy pattern. A natural tendency would be to 
choose seats closer to the centre of the theatre to achieve better viewing. We
attempt to model this with seat choices made via the geometric distribution
with a strong bias towards the centre seat for the central section of the theatre
and for the aisle seat for sections on the sides of the theatre. The results
here are more extreme, in that for $p$ constant, we expect only a constant number
of seats to be occupied when there is a bias towards the entrance of a row while we expect
at least half the row to be filled when the bias is away from the entrance.
In a further
attempt to make the model more realistic we consider  the Zipf distribution on the seat choices, as this 
distribution often arises when considering the cumulative decisions of a
group of humans (though not necessarily Greeks)\cite{zipf}. We show that under this distribution
when theatregoers
are biased towards the entrance to a row, the number of occupied seats is
$\Theta(\ln \ln n)$ while if the bias is towards the centre of the row the number
is $\Theta(\ln^2 n)$.
If we assume that theatregoers proceed to another row if their initial choice
is blocked it is easy to
use our results for single rows with one and two entrances to derive bounds
on the total number of seats occupied in a theatre with multiple rows and aisles.

\subsection{Related work}

Motivation for seating arrangement problems
comes from polymer chemistry and statistical physics in
\cite{flory1939intramolecular,olson1978markov} (see also \cite{strogatz2012joy}[Chapter 19]
for a related discussion).
In particular, the number and size of random independent sets 
on grids (and other graphs) is of great interest in
statistical physics for analyzing {\it hard} particles in lattices
satisfying the exclusion rule, i.e., if a vertex of a lattice is 
occupied by a particle its neighbors must be vacant, and 
have been studied
extensively both in statistical physics and combinatorics 
\cite{baxter,hard1,hard2,calkin-wilf,finch}.

Related to this is the ``unfriendly seating'' arrangement problem 
which
was posed by
Freedman and Shepp \cite{freedman}:
Assume there are $n$ seats in a row at a 
luncheonette and people sit down one at a time at random.
Given that 
they are unfriendly and never sit next to one another,
what is the expected number
of persons to sit down, assuming no moving is allowed?
The resulting density has been studied in
\cite{freedman,friedman,mackenzie} for a $1 \times n$ lattice and 
in \cite{georgiou2009random} for the $2\times n$ and
other lattices.
See also \cite{kk} for a related application to privacy.

Another related problem considers the following natural 
process for generating a maximal independent set of a 
graph~\cite{mitzenmacher}.
Randomly choose a node and place it in the independent set. 
Remove the node and all its neighbors
from the graph. Repeat this process until no nodes remain. 
It is of interest to analyze the expected size of the
resulting maximal independent set. For investigations on a similar
process for generating maximal matchings 
the reader is referred to \cite{aronson1995randomized,dyer1991randomized}.

\subsection{Outline and results of the paper}

We consider the above problem for the case of
a row that has one entrance and the case with two entrances.   We
develop closed form formulas, or almost tight bounds up to multiplicative constants, for the expected
number of occupied seats in a row for any given $n$ and $p$.
First we study  
the simpler problem
for selfish theatregoers, i.e., $p=1$,
in Section~\ref{selfish:sec}. In
Section~\ref{courteous:sec}, we consider $p$-courteous theatregoers.
In these sections, the placement of theatregoers obeys the
uniform distribution.
Section~\ref{geo:sec} considers what happens with $p$-courteous theatregoers under the
geometric distribution.
In Section~\ref{zipf:sec} we look at  theatregoers  whose
placement obeys the Zipf distribution.
And in Section~\ref{theater:sec} we show how the previous results may extended 
to theater arrangements with multiple rows and aisles.
Finally, in Section~\ref{other:sec} we conclude by proposing several open problems
and directions for further research. Details of any missing proofs
can be found in the Appendix.

\section{Selfish Theatregoers}
\label{selfish:sec}

In this section we consider the occupancy problem for a row of seats
arranged next to each other in a line.
First we consider theater occupancy with
selfish theatregoers in that a theatregoer occupying a
seat never gets up to allow another theatregoer to go by.
We consider two types of rows, either
open on one side or open on both sides.
Although the results presented here are easily derived from those 
in Section~\ref{courteous:sec} for the $p$-courteous case,
our purpose here is to introduce the methodology in
a rather simple theatregoer model.


Consider an arrangement of $n$ seats in a row
(depicted in Figure~\ref{fig:th1} as squares).
\begin{figure}[!htb]
\begin{center}
\includegraphics[width=8cm]{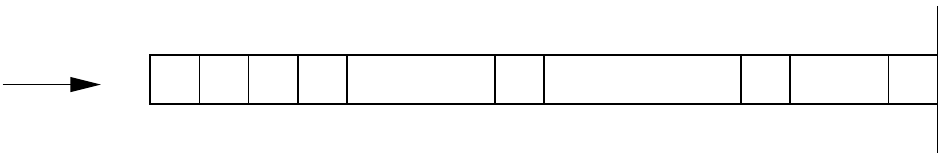}
\end{center}
\caption{An arrangement of seats; theatregoers may enter only from the left
and the numbering of the seats is $1$ to $n$ from left to right.}
\label{fig:th1}
\end{figure}
Theatregoers enter in sequence one after 
the other and may enter the arrangement only from the left.
A theatregoer occupies a seat at random with the uniform distribution
and if selfish (s)he
blocks passage 
to her/his right. What is the expected number 
of occupied seats?


\begin{theorem}[Row with only one entrance]
\label{thm1}
The expected number of occupied seats by selfish theatregoers
in an arrangement of $n$ seats
in a row with single entrance is equal to $H_n$, the $n$th harmonic number.
\end{theorem}
\begin{proof} ({\bf Theorem~\ref{thm1}})
Let $E_n$ be the expected number of theatregoers occupying seats 
in a row of $n$ seats.
Observe that $E_0=0, E_1 =1$ and that the following recurrence is valid
for all $n \geq 1$.
\begin{eqnarray}
E_n &=& \label{maineq1}
1 + \frac{1}{n} \sum_{k=1}^{n} E_{k-1} 
= 1 + \frac{1}{n} \sum_{k=1}^{n-1} E_k.
\end{eqnarray}
The explanation for this equation is as follows. A theatregoer
may occupy any one of the seats from $1$ to $n$. If it
occupies seat number $k$ then seats numbered $k+1$ to $n$
are blocked while only seats numbered $1$ to $k-1$ may be
occupied by new theatregoers. 
It is not difficult to solve this recurrence. Write down
both recurrences for $E_n$ and $E_{n-1}$.
\begin{equation*}
nE_n 
=
n + \sum_{k=1}^{n-1} E_k 
\mbox{ and } 
(n-1)E_{n-1}
= n-1 + \sum_{k=1}^{n-2} E_k.
\end{equation*}
Substracting these two identities we see that
$nE_n - (n-1) E_{n-1} = 1 + E_{n-1}$.
Therefore $E_n = \frac{1}{n} + E_{n-1}$.
This proves Theorem~\ref{thm1}.
\qed
\end{proof}


Now consider an arrangement of $n$ seats
(depicted in Figure~\ref{fig:th2})
\begin{figure}[!htb]
\begin{center}
\includegraphics[width=10cm]{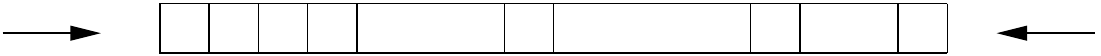}
\end{center}
\caption{An arrangement of $n$ seats; theatregoers may enter either from the right
or from the left.}
\label{fig:th2}
\end{figure}
with two entrances
such that 
theatregoers may enter only from either right or left.
In what follows, we invoke several times the approximate size of 
the harmonic number $H_n$ 
which can be expressed as follows
$$
H_n = \ln n + \gamma + \frac{1}{2n} + o (n),
$$
where $\gamma$ is Euler's constant~\cite{knuth}.

\begin{theorem}[Row with two entrances]
\label{thm1two}
The expected number of occupied seats by selfish theatregoers
in an arrangement of $n$ seats
in a row with two entrances is $2 \ln n$, 
asymptotically in $n$.
\end{theorem}

%

\begin{proof}({\bf Theorem~\ref{thm1two}})
Let $F_n$ be the expected number of occupied seats
in a line with two entrances and $n$ seats.
Further, let $E_n$ be the expected number of theatregoers occupying seats 
in a line with a single entrance and $n$ seats, which is
the function defined
in the proof Theorem~\ref{thm1}.

Observe that 
\begin{eqnarray}
F_n &=& \label{maineq2}
1 + \frac{1}{n} \sum_{k=1}^n (E_{k-1} + E_{n-k})
\end{eqnarray}
The explanation for this is as follows. 
The first theatregoer may occupy any position
$k$ in the row of $n$ seats.
Being selfish, entry is possible only 
from one side of the row, i.e., the next seat that
can be occupied is numbered either
from $1$ to $k-1$ or from $k+1$ to $n$.  

It follows from Theorem~\ref{thm1} and using the standard approximation
for the harmonic number (see~\cite{knuth}) that
\begin{eqnarray*}
F_n &=&
1 + \frac{1}{n} \sum_{k=1}^n (H_{k-1} + H_{n-k})
=
1 + \frac{2}{n} \sum_{k=1}^n H_{k-1} 
=
2 \ln n + O(1),
\end{eqnarray*}
which proves Theorem~\ref{thm1two}.
\qed
\end{proof}

\section{Courteous Theatregoers}
\label{courteous:sec}

Now consider the case where theatregoers are 
courteous with probability $p$ and
selfish with
probability $1-p$. 
We assume that the probabilistic behaviour of
the theatregoers is independent of each other and
it is set at the start and remains the same throughout
the occupancy of the row of seats. Analysis
of the occupancy will be done separately
for rows of seats with one and two entrances
(see Figures~\ref{fig:th1}~and~\ref{fig:th2}). Again, 
seat choices are made uniformly at random. 
Observe that
for $p=1$ no theatregoer is selfish and therefore all
seats in a row of seats will be occupied.
Also, since
the case $p=0$ whereby all theatregoers are selfish
was analyzed in the last section, we can assume without
loss of generality
that $0 < p < 1$.


\begin{theorem}[Row with only one entrance]
\label{thm01p}
Assume $0 < p < 1$ is given.
The
expected number $E_n$ of occupied seats in an arrangement of $n$ seats
in a row having only one entrance at an endpoint
with $p$-courteous
theatregoers is given by the expression
\begin{equation}
\label{pach1}
E_n = \sum_{k=1}^n \frac{1-p^k}{k(1-p)},
\end{equation}
for $n \geq 1$. 
In particular, for fixed $p$, $E_n$ is $\frac{H_n + \ln (1-p)}{1-p}$,
asymptotically in $n$.
\end{theorem}
\begin{proof} ({\bf Theorem~\ref{thm01p}})
Consider an arrangement of $n$ seats (depicted in Figure~\ref{fig:th1} as squares).
Let $E_n$ denote the expected number of occupied
positions  in an arrangement of $n$ seats
with single entrance 
at an endpoint
and $p$-courteous
theatregoers.
With this definition in mind we obtain the following
recurrence
\begin{eqnarray}
E_n 
&=& \label{rec01}
1 + p E_{n-1} 
+  \frac{1-p}{n} \sum_{k=1}^n E_{k-1}
\end{eqnarray}
where the initial condition $E_0 =0$ holds.

Justification for this recurrence is as follows.
Recall that we have a line
with single entrance on the left. Observe that with probability $1-p$
the theatregoer is selfish and if (s)he occupies
position $k$ then theatregoers arriving later
can only occupy a position in the interval $[1, k-1]$
with single entrance at $1$.
On the other hand, with probability $p$ the theatregoer is courteous
in which case the next person
arriving sees $n-1$ available seats as 
far as (s)he is concerned; where the first person
sat doesn't matter and what remains is a problem of size $n-1$.
This
yields the desired recurrence.

To simplify, multiply Recurrence~\eqref{rec01} by $n$
and combine similar terms to derive
\begin{eqnarray}
nE_n 
&=& \notag
n + (np+1-p)E_{n-1} + (1-p) \sum_{k=1}^{n-2} E_k .
\end{eqnarray}
A similar equation is obtained when we replace $n$ with $n-1$
\begin{eqnarray}
(n-1)E_{n-1} 
&=& \notag
n-1 + ((n-1)p+1-p)E_{n-2} + (1-p) \sum_{k=1}^{n-3} E_k .
\end{eqnarray}
If we substract these last two equations
we derive
$
nE_n -(n-1)E_{n-1}=
1 + (np+1-p)E_{n-1} - ((n-1)p+1-p)E_{n-2}
+(1-p) E_{n-2} .
$
After collecting similar terms. it follows that 
$nE_n = 1+ (n(1+p) -p) E_{n-1} -(n-1)p E_{n-2}$.

Dividing both sides of the last equation
by $n$ we obtain the following recurrence
\begin{eqnarray}
E_n 
&=& \notag
\frac{1}{n} + 
\left( 1 + p - \frac{p}{n} \right) E_{n-1}
-\left(1-\frac{1}{n} \right)pE_{n-2},
\end{eqnarray}
where it follows easily from the occupancy conditions that
$E_0 = 0, E_1 = 1, E_2 = \frac{3}{2} + \frac{p}{2}$. 
Finally,
if we define $D_n := E_n - E_{n-1}$, substitute
in the last formula and collect similar terms
we conclude that
\begin{eqnarray}
D_n
&=& \label{rec02}
\frac{1}{n} + 
\left(1-\frac{1}{n} \right)p D_{n-1},
\end{eqnarray}
where $D_1 = 1$. The solution of Recurrence~\eqref{rec02}
is easily shown to be $D_n = \frac{1-p^n}{n(1-p)}$ for $p<1$.
By telescoping we have the identity $E_n = \sum_{k=1}^n D_k $. The proof of the theorem is complete once we observe that $\sum_{k=1}^\infty p^k/k = - \ln (1-p)$.
\qed
\end{proof}

\begin{theorem}[Row with two entrances]
\label{thm02p}
Assume $0 < p < 1$ is given.
The
expected number $F_n$ of occupied seats in an arrangement of $n$ seats
in a row having two entrances at the endpoints
with probabilistically $p$-courteous
theatregoers is given by the expression
\begin{equation}
\label{pach2}
F_n =
-\frac{1-p^n}{1-p} + 2 \sum_{k=1}^n \frac{1-p^k}{k(1-p)},
\end{equation}
for $n \geq 1$. 
In particular, for fixed $p$, $F_n$ is $-\frac{1}{1-p} + 2\frac{H_n - \ln (1-p)}{1-p},$
asymptotically in $n$.
\end{theorem}

\begin{proof} ({\bf Theorem~\ref{thm02p}})
Now consider an arrangement of $n$ seats
(depicted in Figure~\ref{fig:th2}).
For fixed $p$,
let $F_n$ denote the expected number of occupied
positions  in an arrangement of $n$ seats
with two entrances one at each endpoint 
and probabilistically $p$-courteous
theatregoers.
Let $E_n$ denote the expected number of occupied
positions  in an arrangement of $n$ seats
with single entrance 
and probabilistically $p$-courteous
theatregoers (defined in Theorem~\ref{thm01p}).
With this definition in mind we obtain the following
recurrence
\begin{eqnarray}
F_n 
&=& \label{rec002}
1 + p F_{n-1} 
+  \frac{1-p}{n} \sum_{k=1}^n (E_{n-k} + E_{k-1})
\end{eqnarray}
where the initial conditions $E_0 = F_0 =0$ hold.

Justification for this recurrence is as follows.
We have a line
with both entrances at the endpoints open. 
Observe that with probability $1-p$
the theatregoer is selfish and if it occupies
position $k$ then theatregoers arriving later
can occupy positions in $[1, k-1] \cup [k+1, n]$
such that in the interval $[1, k-1]$
a single entrance is open at $1$ and 
in the interval $[k+1, n]$ a single
entrance is open at $n$.
On the
other hand, like in the single entrance case,
with probability $p$ the theatregoer is courteous in which case the next person
arriving sees $n - 1$ 
available seats as far as (s)he is concerned; where the first
person sat doesn't matter.
This yields the desired recurrence.

Using Equation~\eqref{rec01},
it is clear that Equation~\eqref{rec002} can be simplified to
\begin{eqnarray}
F_n 
&=& \notag 
1 + p F_{n-1} 
+  \frac{2(1-p)}{n} \sum_{k=1}^n E_{k-1}
= \notag 
1 + p F_{n-1} 
+  2(E_n -1 -p E_{n-1}),
\end{eqnarray}
which yields
\begin{eqnarray}
F_n -1 - p F_{n-1} 
&=&
\label{rec0023} 
2(E_n -1 -p E_{n-1}) 
\end{eqnarray}

Finally if we define $\Delta_n := F_n -2E_n$ then Equation~\eqref{rec0023} 
gives rise to the following recurrence
\begin{equation}
\label{rec0024} 
\Delta_n = -1 + p \Delta_{n-1},
\end{equation}
with initial condition $\Delta_1 = F_1 - 2E_1 = - 1$.
Solving Recurrence~\eqref{rec0024} we conclude that
$\Delta_n = -\frac{1-p^n}{1-p}$, for $p<1$, and $\Delta_n =-n$,
otherwise. Therefore,
$
F_n = \Delta_n + 2E_n ,
$
from which we derive the desired Formula~\eqref{pach2}. 
Using the expansion of $\ln (1-p)$ in a Taylor series (in the variable $p$)
we get the claimed expression for fixed $p$
and conclude the proof
of the theorem.
\qed
\end{proof}

\section{Geometric Distribution}
\label{geo:sec}

In the sections above the theatregoers were more or less oblivious
to the seat they selected in that they chose their
seat independently at random with the uniform distribution. A
more realistic assumption might be that theatregoers prefer to be
seated as close to the centre of the action as possible. For a row
in the centre of the theatre, this suggests that there would be
a bias towards the centre seat (or two centre seats in the case of an even
length row) which is nicely modelled by a row with one entrance ending
at the middle of the row
where the probability of choosing a seat is biased towards the centre seat (which
we consider to be a barrier, i.e., people never go past the centre if they enter
on a given side of a two sided row).  
For a row towards the edge of the theatre this would imply that
theatregoers prefer to chose their seats as close to the aisle, i.e.,
as close to the entrance, as possible. This is nicely modelled by
a row with one entrance with a bias towards the entrance. 

As usual, we consider a row with one entrance with $n$ seats
(depicted in Figure~\ref{fig:th1} as squares) 
numbered $1, 2, \ldots n$ from left to right.
We
refer to a distribution modelling the first case, with bias away from the entrance, as a
distribution with a {\em right} bias, while in the second case, with bias towards
the entrance, as distribution with a {\em left} bias. (We only consider cases where
the bias is monotonic in one direction though one could consider more 
complicated distributions if for example there are obstructions part of the way 
along the row.)

A very strong bias towards the centre might be modelled by the geometric distribution. 
For the case of a left biased distribution theatregoers will occupy seat $k$ with probability $\frac{1}{2^k}$
for $k=1, \ldots, n-1$ and with probability $\frac{1}{2^{n-1}}$ for $k=n$. 
For the case of a right biased distribution theatregoers will occupy seat
$k$ with probability $\frac{1}{2^{n+1-k}}$ for $k= 2, \ldots, n$ and with probability $\frac{1}{2^{n-1}}$ for $k=1$.
We examine the occupancy of a one-entrance row under each of these distributions assuming 
a $p$-courteous audience.


\begin{theorem}[Left bias]
\label{thm1geol}
The expected number of occupied seats by $p$-courteous 
theatregoers in an arrangement of $n$ seats
in a row with single entrance is 
\begin{equation}
\label{geo1:eq}
\sum_{l=1}^n \prod_{k=1}^{l-1} \left( p + \frac{1-p}{2^{k}} \right) 
\end{equation}
In particular, the value $T_p$ of ~\eqref{geo1:eq} as $n\rightarrow \infty$, satisfies 
$$
\frac{1.6396 -0.6425 p}{1-p}
\leq T_p \leq 
\frac{1.7096 -0.6425 p}{1-p}
$$
for all $p<1$.
\end{theorem}

\begin{proof} ({\bf Theorem~\ref{thm1geol}})
In the geometric distribution with left bias a theatergoer
occupies seat numbered $k$ with probability
$2^{-k}$, for $k \leq n-1$ and seat numbered $n$
with probability $2^{-(n-1)}$. 
The seat occupancy
recurrence for courteous theatergoers is the following
\begin{eqnarray}
L_n &=& \label{lgeo1:eq}
1 + p L_{n-1} + (1-p) \sum_{k=1}^{n-1} 2^{-k}L_{k-1} + (1-p)2^{-(n-1)}L_{n-1}
\end{eqnarray}
with initial condition $L_0=0, L_1 = 1$.
To solve this recurrence we consider the expression for $L_{n-1}$
\begin{eqnarray}
L_{n-1} &=& \label{lgeo2:eq}
1 + p L_{n-2} + (1-p) \sum_{k=1}^{n-2} 2^{-k}L_{k-1} + (1-p)2^{-(n-2)}L_{n-2}
\end{eqnarray}
Subtracting Equation~\eqref{lgeo2:eq} from Equation~\eqref{lgeo1:eq}
and using the notation $\Delta_k := L_k - L_{k-1}$ we see that
\begin{eqnarray}
\Delta_n &=& \notag
\left( p + \frac{1-p}{2^{n-1}} \right) \Delta_{n-1} , 
\end{eqnarray}
for $n \geq 2$.
It follows that
\begin{eqnarray}
\Delta_n &=& \notag
\prod_{k=1}^{n-1} \left( p + \frac{1-p}{2^{k}} \right)  , 
\end{eqnarray}
which proves Identity~\eqref{geo1:eq}.

The previous identity implies that
$\Delta_n \leq \left(\frac{1+p}{2}\right)^{n-1}$ and therefore
we can get easily an upper bound on the magnitude of
$L_n$. Indeed,
\begin{eqnarray}
L_n &=& \notag
\sum_{k=1}^n \Delta_k 
\leq
\sum_{k=1}^{n-1} \left( \frac{1+p}{2} \right)^{k-1}
\leq
\frac{2}{1-p} ,
\end{eqnarray}
for $p<1$. Similarly, one can easily show a lower bound of $1/(1-p)$. Next we focus on showing the much tighter bounds we have already promised. 

Our goal is to provide good estimates of $T_p = \sum_{l=1}^\infty \prod_{k=1}^{l-1} \left( p + \frac{1-p}{2^{k}} \right)$. Although there seems to be no easy closed formula that describes $T_p$, the same quantity can be numerically evaluated for every fixed value of $p<1$ using any mathematical software that performs symbolic calculations. In particular we can draw $T_p$ for all non negative values of $p<1$. 

One strategy to approximate $T_p$ to a good precision would be to compute enough points $(p,T_p)$, and then find an interpolating polynomial. Since we know $T_p$ is unbounded as $p\rightarrow 1^-$, it seems more convenient to find interpolating points $(p,(1-p)T_p)$ instead (after all, we know that $1/(1-p) \leq T_p \leq 2/(1-p)$). Adding at the end a sufficient error constant term, we can find polynomials that actually bound from below and above expression $(1-p)T_p$. 

It turns out that just a few interpolating points are enough to provide a good enough estimate. In that direction, we define polynomial 
$$
g(p) :=1.6746 -0.6425 p
$$
which we would like to show that approximates $(1-p)T_p$ sufficiently well. 
\begin{figure}[!htb]
\begin{center}
\includegraphics[width=8cm]{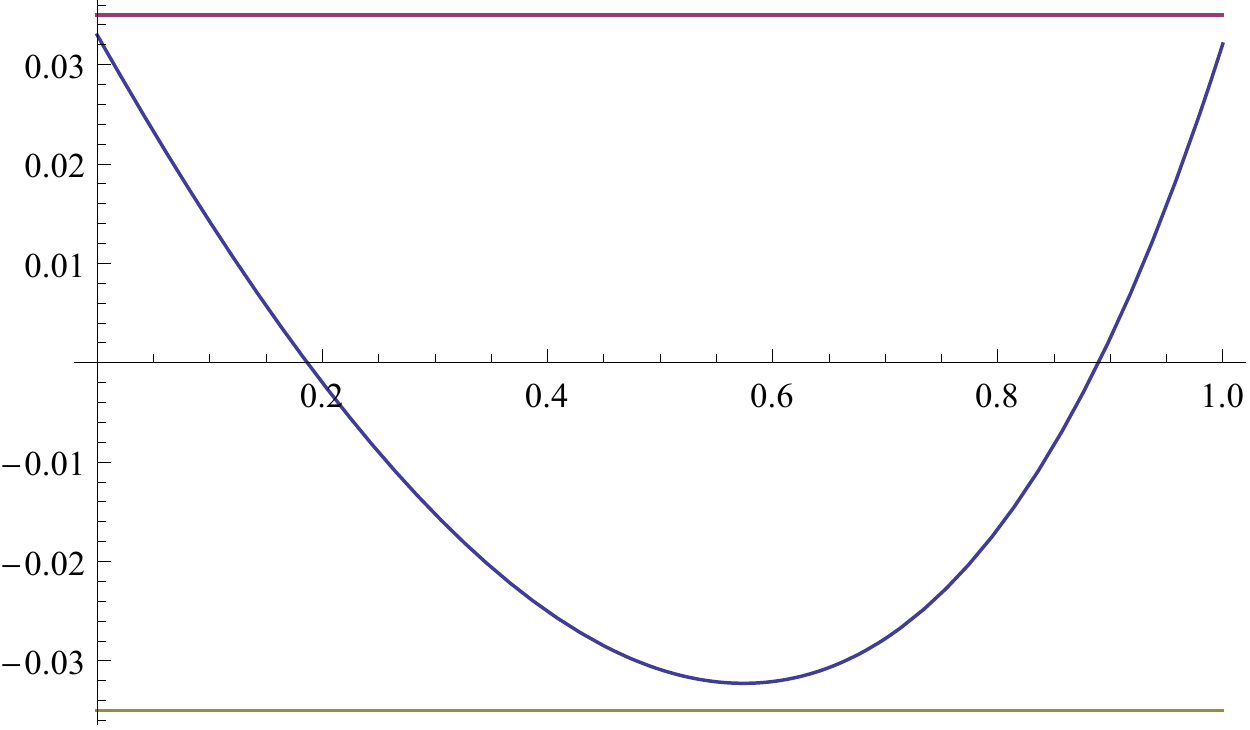}
\caption{The graph of $g(p) - (1-p)T_p$ together with the bounds $\pm0.035$.}
\label{fig: interpolation}
\end{center}
\end{figure}
To that end, we can draw $g(p) - (1-p)T_p$, see Figure~\ref{fig: interpolation}, and verify that indeed $\left| g(p) - (1-p)T_p \right|\leq 0.035$ as promised. 
\qed
\end{proof}

We leave it as an open problem to determine the exact asymptotics of expression~\eqref{geo1:eq} above, as a function of $p$. As a sanity check, we can find (using any mathematical software that performs symbolic calculations) the limit of ~\eqref{geo1:eq} as $n\rightarrow \infty$ when $p=0$, which turns out to be approximately $1.64163$. 


\begin{theorem}[Right bias]
\label{thm1geor}
The expected number of occupied seats by $p$-courteous 
theatregoers in an arrangement of $n$ seats
in a row with single entrance is at least $\frac{n+1}{2}$, 
for any $p$. Moreover, this bound is attained for $p=0$.
\end{theorem}

\begin{proof} ({\bf Theorem~\ref{thm1geor}})
In the geometric distribution with left bias a theatergoer
occupies seat numbered $k$ with probability
$2^{k-n-1}$, for $k \geq 2$ and seat numbered $1$
with probability $2^{-(n-1)}$.
The seat occupancy
recurrence for courteous theatergoers is the following
\begin{eqnarray}
R_n &=& \label{rgeo1:eq}
1 + p R_{n-1} + (1-p) \sum_{k=2}^{n} 2^{n+1-k}R_{k-1} .
\end{eqnarray}
with initial condition $R_0=0, R_1 = 1$.
To solve this recurrence, we
consider as usual the equation for $R_{n-1}$
\begin{eqnarray}
R_{n-1} &=& \label{rgeo2:eq}
1 + p R_{n-2} + (1-p) \sum_{k=2}^{n-1} 2^{n-k}R_{k-1} .
\end{eqnarray}
Subtracting Equation~\eqref{rgeo2:eq} from Equation~\eqref{rgeo1:eq}
and using the notation $\Delta_k := R_k - R_{k-1}$ we see that
\begin{eqnarray}
\Delta_n &=& \label{rgeo3:eq}
p \Delta_{n-1} + (1-p) \sum_{k=1}^{n-2} \frac{1}{2^k} \Delta_{n-k} + \frac{1-p}{2^{n-1}}, 
\end{eqnarray}
for $n \geq 2$. We claim that we
can use Equation~\eqref{rgeo3:eq} to prove that
$\Delta_k \geq \frac{1}{2}$ by induction on $k$, for all
$k \geq 1$. 
Observe $\Delta_1 = 1$.
Assume the claim is valid for all $1 \leq k \leq n-1$. Then we see that
\begin{eqnarray}
\Delta_n 
&\geq& \notag
\frac{p}{2} +
\frac{1-p}{2} \left( 1- \frac{1}{2^{n-2}} \right) + \frac{1-p}{2^{n-1}}
=
\frac{1}{2},
\end{eqnarray}
which proves the claim.

It is easy to see that this same proof can be used to
show that $\Delta_n=\frac1 2$, for all $n \geq 2$, in the
case $p=0$. This proves the theorem.
\qed
\end{proof}

\section{Zipf Distribution}
\label{zipf:sec}

We now study the case where theatregoers
select their seat 
using an arguably more natural  distribution, namely, the Zipf distribution \cite{zipf}.
As before, throughout the presentation we
consider an arrangement of $n$ seats 
(depicted in Figure~\ref{fig:th1} as squares) 
numbered $1$ to $n$ from left to right with
one entrance starting from seat $1$. 
Theatregoers enter in sequentially 
and may enter the row only from the single entrance.
There are two occupancy possibilities: {\em Zipf with left bias} and
{\em Zipf with right bias}.
In Zipf with left bias (respectively, right) a 
theatregoer will occupy 
seat $k$ at random with probability $\frac{1}{kH_n}$
(respectively, $\frac{1}{(n+1-k)H_n}$)
and a selfish theatregoer
blocks passage 
to her/his right, i.e., all positions in 
$[k+1,n]$. 
In the sequel we look at a row with a single
entrance. The case of a row with two entrances
may be analyzed in a similar manner.


First we analyze the Zipf distribution 
with left bias for selfish theatregoers. 

\begin{theorem}[Selfish with left bias]
\label{thm1zipfl}
The expected number of occupied seats by selfish theatregoers
in an arrangement of $n$ seats
in a row with single entrance is equal to $\ln \ln n$, asymptotically in $n$.
\end{theorem}
\begin{proof} ({\bf Theorem~\ref{thm1zipfl}})
Let $L_n$ be the expected number of theatregoers occupying seats 
in a row of $n$ seats.
Observe that $L_0=0, L_1 =1$ and that the following recurrence is valid
for all $n \geq 1$.
\begin{eqnarray}
L_n &=& \label{maineq1z}
= 1 + \frac{1}{H_n} \sum_{k=1}^{n} \frac{1}{k} L_{k-1}.
\end{eqnarray}
The explanation for this equation is as follows. A theatregoer
may occupy any one of the seats from $1$ to $n$. If it
occupies seat number $k$ then seats numbered $k+1$ to $n$
are blocked while only seats numbered $1$ to $k-1$ may be
occupied by new theatregoers. 

It is not difficult to solve this recurrence. Write down
both recurrences for $L_n$ and $L_{n-1}$.
\begin{eqnarray*}
H_n L_n 
&=& 
H_n + \sum_{k=1}^{n} \frac{1}{k} L_{k-1} 
\mbox{ and }
H_{n-1} L_{n-1} 
= H_{n-1} + \sum_{k=1}^{n-1} \frac{1}{k} L_{k-1}.
\end{eqnarray*}
Substracting these last two identities we see that
$$
H_nL_n - H_{n-1} L_{n-1} = H_n - H_{n-1} + \frac{1}{n} L_{n-1}
= \frac{1}{n} + \frac{1}{n} L_{n-1}
$$
Therefore $H_n L_n = \frac{1}{n} + H_n L_{n-1}$.
Consequently,
$
L_n = \frac{1}{nH_n} + L_{n-1}.
$
From the last equation we see that
$$
L_n = \sum_{k=2}^n \frac{1}{kH_k} \approx \int_2^n \frac{dx}{x \ln x}  = \ln \ln n .
$$
This yields easily Theorem~\ref{thm1zipfl}.
\qed
\end{proof}


Next we consider  selfish theatregoers choosing their seats according to 
the Zipf distribution 
with right bias. As it turns out, the analysis of the resulting recurrence is
more difficult than the previous cases. First we need the following technical lemma whose proof can be found in the Appendix.
\begin{lemma}\label{lem: both bounds on summation}
For every $\epsilon>0$, there exists $n_0$ big enough such that
$$
\left| \frac{\pi^2}6 - \sum_{k=1}^{n-1}\frac{H_n-H_k}{n-k} \right| \leq \epsilon, \quad \forall n\geq n_0
$$
In particular, for all $n \geq 40$ we have 
$$
1.408 \leq \sum_{k=1}^{n-1}\frac{H_n-H_k}{n-k} \leq 1.86.
$$
\end{lemma}
Next we use Lemma~\ref{lem: both bounds on summation} to conclude that
\begin{lemma}
\label{costis2}
The solution of the recurrence relation 
$$R_n = 1 + \frac1{H_n}\sum_{k=1}^{n-1}\frac{1}{n-k}R_k$$
with initial condition $R_1=1$ satisfies 
\begin{equation}\label{equa: both bounds}
\frac{100}{383} H^2_n \leq R_n \leq  \frac{5}{7} H^2_n.
\end{equation}
\end{lemma}

\begin{proof} ({\bf Lemma~\ref{costis2}})
It is easy to check numerically that for $n_0=40$ we have
$$
\frac{R_{n_0}}{H^2_{n_0}} \approx 0.430593
$$
and indeed $\frac{100}{383} \leq 0.430593 \leq \frac{5}{7}$.

Hence, the promised bounds follow inductively on $n\geq n_0$, once we prove that for the constants $c'=\frac57,c''=\frac5{19}$ and that for all $n\geq n_0$ we have
\begin{align*}
1 + \frac {c'}{H_n}\sum_{k=1}^{n-1}\frac{1}{n-k}H^2_k 
&\leq {c'} H^2_n \\
1 + \frac{c''}{H_n}\sum_{k=1}^{n-1}\frac{1}{n-k}H^2_k 
&\geq c'' H^2_n \\
\end{align*}

To save repetitions in calculations, let $\Box \in \{\leq, \geq\}$ and $c \in \{c',c''\}$, and observe that 
\begin{align}
1 + \frac {c}{H_n}\sum_{k=1}^{n-1}\frac{1}{n-k}H^2_k ~~\Box~~ c H^2_n 
~~~&\Leftrightarrow ~~~
\frac{H_n}{c} ~~\Box~~ H^3_n -  \sum_{k=1}^{n-1}\frac{1}{n-k}H^2_k \notag \\
~~~&\Leftrightarrow ~~~
\frac{H_n}{c} ~~\Box~~ \sum_{k=1}^{n-1}\frac{H^2_n-H^2_k}{n-k} + \frac{H_n^2}{n} \tag{Since $\sum_{k=0}^{n-1}\frac1{n-k}=H_n$}\\
~~~&\Leftrightarrow ~~~
\frac{1}{c} - \frac{H_n}n ~~\Box~~ \frac1{H_n}\sum_{k=1}^{n-1}\frac{H^2_n-H^2_k}{n-k} \notag \\
~~~&\Leftrightarrow ~~~
\frac{1}{c} - \frac{H_n}n ~~\Box~~ \frac1{H_n}\sum_{k=1}^{n-1}\frac{(H_n+H_k)(H_n-H_k)}{n-k} \label{equa: aimed inequalities} \end{align}

For proving the upper bound of~\eqref{equa: both bounds}, we use $\Box="\leq"$ (note that the direction is inversed). We focus on expression \eqref{equa: aimed inequalities} which we need to show that is satisfied for the given constant. In that direction we have
$$
\frac1{H_n}\sum_{k=1}^{n-1}\frac{(H_n+H_k)(H_n-H_k)}{n-k} 
\geq
\sum_{k=1}^{n-1}\frac{H_n-H_k}{n-k} 
\stackrel{(Lemma~\ref{lem: both bounds on summation})}{\geq} 1.408 \geq \frac75-\frac{H_n}n
$$
Hence, \eqref{equa: aimed inequalities} is indeed satisfied for $c=\frac57$, establishing  the upper bound of~\eqref{equa: both bounds}.

Now for the lower bound of~\eqref{equa: both bounds}, we take $\Box="\geq"$ , and we have 
$$
\frac1{H_n}\sum_{k=1}^{n-1}\frac{(H_n+H_k)(H_n-H_k)}{n-k} 
\leq
2 \sum_{k=1}^{n-1}\frac{H_n-H_k}{n-k} 
\stackrel{(Lemma~\ref{lem: both bounds on summation})}{\leq} 3.72 \leq \frac{383}{100}-\frac{H_n}n
$$
for $n\geq 40$. Hence $c''=\frac{100}{383}$, again as promised. 
\qed
\end{proof}

Note that Lemma~\ref{costis2} 
implies that $\lim_{n\rightarrow \infty} R_n/\ln^2 n =c$, 
for some constant $c \in [0.261, 0.72]$. This is actually the constant hidden in the $\Theta$-notation of Theorem~\ref{thm1zipfr}. We leave it as an open problem to determine exactly the constant $c$. Something worthwhile noticing is that our arguments cannot narrow down the interval of that constant to anything better than $[3/\pi^2, 6/\pi^2]$.

\begin{theorem}[Selfish with right bias]
\label{thm1zipfr}
The expected number of occupied seats by selfish theatregoers
in an arrangement of $n$ seats
in a row with single entrance is $\Theta( \ln^2 n)$, 
asymptotically in $n$.
\end{theorem}
\begin{proof} ({\bf Theorem~\ref{thm1zipfr}})
Let $R_n$ be the expected number of theatregoers occupying seats 
in a row of $n$ seats,
when seating is biased to the right,
Observe that $R_0=0, R_1 =1$ and that the following recurrence is valid
for all $n \geq 1$.
\begin{eqnarray}
R_n 
&=& \label{maineq1zz} 
1 + \frac{1}{H_n} \sum_{k=2}^{n} \frac{1}{n+1-k} R_{k-1}
= 
1 + \frac{1}{H_n} \sum_{k=1}^{n-1} \frac{1}{n-k} R_{k}.
\end{eqnarray}
The justification for the recurrence is the same as in the case of the left bias with
the probability changed to reflect the right bias. The theorem now follows immediately
from Lemma~\ref{costis2}. 
\qed
\end{proof}

%


\begin{theorem}[Courteous with left bias]
\label{thm2zipfl}
The expected number of occupied seats by $p$-courteous theatregoers
in an arrangement of $n$ seats
in a row with single entrance is equal to
\begin{eqnarray}
L_n 
&=& \label{maineq5zipf}
\ln \ln n +
\sum_{l=1}^n
\sum_{k=1}^l p^k 
\left( 1 - h_l \right)
\left( 1 - h_{l-1} \right)
\cdots
\left( 1 - h_{l-k+1} \right)
h_{l-k}
\end{eqnarray}
asymptotically in $n$, where $h_0 := 0$ and 
$h_k := \frac{1}{kH_k}$, for $k \geq 1$. In particular, for constant
$0< p < 1$ we have that $L_n = \Theta (\frac{\ln\ln n}{1-p})$.
\end{theorem}

\begin{proof} ({\bf Theorem~\ref{thm2zipfl}})
We obtain easily the following recurrence
\begin{eqnarray}
L_n &=& \label{maineq2zipf}
1 + p L_{n-1} + \frac{1-p}{H_n} \sum_{k=1}^{n} \frac{1}{k} L_{k-1}.
\end{eqnarray}
Write the recurrence for $L_{n-1}$:
\begin{eqnarray}
L_{n-1} &=& \notag
1 + p L_{n-2} + \frac{1-p}{H_{n-1}} \sum_{k=1}^{n-1} \frac{1}{k} L_{k-1}.
\end{eqnarray}

Multiply these last two recurrences by $H_n, H_{n-1}$ respectively to get
\begin{eqnarray}
H_n L_n 
&=& \notag
H_n + p H_n L_{n-1} + (1-p) \sum_{k=1}^{n} \frac{1}{k} L_{k-1} \\
H_{n-1} L_{n-1} 
&=& \notag
H_{n-1} + p H_{n-1} L_{n-2} + (1-p) \sum_{k=1}^{n-1} \frac{1}{k} L_{k-1}
\end{eqnarray}

Now subtract the second equation from the first and after collecting
similar terms and simplifications
we get
\begin{eqnarray}
L_n 
&=& \notag
\frac{1}{nH_n} + \left( 1+p - \frac{p}{nH_n} \right) L_{n-1} 
-p \frac{H_{n-1}}{H_n} L_{n-2} ,
\end{eqnarray}
with initial conditions $L_0 =0, L_1 = 1$.
In turn, if we set $\Delta_n := L_n - L_{n-1}$ then we 
derive the following recurrence for $\Delta_n$.
\begin{eqnarray}
\Delta_n 
&=& \label{maineq3zipf}
\frac{1}{nH_n} + p \left( 1 - \frac{1}{nH_n} \right) \Delta_{n-1},
\end{eqnarray}
with initial condition $\Delta_1 =1$.
Recurrence~\eqref{maineq3zipf} gives rise
to the following expression for $\Delta_n$
\begin{eqnarray}
\Delta_n 
&=& \label{maineq4zipf}
h_n + 
\sum_{k=1}^n p^k 
\left( 1 - h_n \right)
\left( 1 - h_{n-1} \right)
\cdots
\left( 1 - h_{n-k+1} \right)
h_{n-k},
\end{eqnarray}
where $h_0 := 0$ and $h_k := \frac{1}{kH_k}$.
This completes the proof of Identity~\eqref{maineq5zipf}.

Next we prove the bounds on $L_n$. 
First of all observe that the following inequality holds
\begin{equation}
\label{hn:eq}
2 h_n \leq h_{n/2} \leq 3 h_n .
\end{equation}

Next we estimate the sum 
in the righthand side of Equation~\eqref{maineq4zipf}.
To this end we split the sum 
into two parts: one part, say $S_1$, in the range from $1$ to $n/2$ and
the second part, say $S_2$, from $n/2 +1$ to $n$.
Observe that
\begin{eqnarray}
S_2
&=& \notag
\sum_{k\geq n/2+1} p^k  \left( 1 - h_n \right) \left( 1 - h_{n-1} \right) \cdots \left( 1 - h_{n-k+1} \right) h_{n-k}\\
&\leq& \notag
\sum_{k\geq n/2+1} p^k \leq p^{n/2+1} \frac{1}{1-p} ,
\end{eqnarray}
which is small, asymptotically in $n$, for $p<1$ constant. 

Now consider the sum $S_1$.
\begin{eqnarray}
S_1
&=& \notag
\sum_{k=1}^{n/2} p^k  \left( 1 - h_n \right) \left( 1 - h_{n-1} \right) \cdots \left( 1 - h_{n-k+1} \right) h_{n-k}\\
&\leq& \notag
h_{n/2} \sum_{k=1}^{n/2} p^k \leq 3 h_n \frac{p}{1-p} \mbox{ (Using Inequality~\eqref{hn:eq})} 
\end{eqnarray}
and
\begin{eqnarray}
S_1
&\geq& \notag
h_n \sum_{k=1}^{n/2} p^k  \left( 1 - h_n \right) \left( 1 - h_{n-1} \right) \cdots \left( 1 - h_{n-k+1} \right) \\
&\approx& \notag
h_n \sum_{k=1}^{n/2} p^k  e^{-(h_n +h_{n-1}+ \cdots + h_{n-k+1})} \mbox{ (since $1-x \approx e^{-x}$)} \\
&\approx& \notag
h_n \sum_{k=1}^{n/2} p^k  e^{- \ln \left( \frac{\ln n}{\ln (n/2)} \right)} 
\approx
h_n \sum_{k=1}^{n/2} p^k 
\approx
c h_n \frac{p}{1-p},
\end{eqnarray}
for some constant $c>0$.
Combining the last two inequalities it is easy to derive tight
bounds for $\Delta_n$ and also for $L_n$, since
$L_n = \sum_{k=1}^n \Delta_k$.
This completes the proof of Theorem~\ref{thm2zipfl}.
\qed
\end{proof}

\begin{theorem}[Courteous with right bias]
\label{thm2zipfr}
The expected number $R_n(p)$ of occupied seats by $p$-courteous theatregoers
in an arrangement of $n$ seats
in a row with single entrance, and for all constants $0 \leq p<1$ satisfies 
$$R_n(p) = \Omega \left( \frac{H_n^2}{1 - 0.944p} \right)
\mbox{ and } R_n(p) = O\left( \frac{H_n^2}{1-p} \right)$$
asymptotically in $n$.
\end{theorem}

\begin{proof} ({\bf Theorem~\ref{thm2zipfr}})
Let $R_n(p)$ be the expected number of theatregoers occupying seats 
in a row of $n$ seats,
when seating is biased to the right.
Observe that $R_0(p)=0, R_1(p) =1$ and that the following recurrence is valid
for all $n \geq 1$.
\begin{eqnarray}
R_n(p) 
&=& \notag
1 + p R_{n-1}(p) + \frac{1-p}{H_n} \sum_{k=1}^{n} \frac{1}{n+1-k} R_{k-1}(p)\\
&=&  \label{maineq2zzipf} 
1 + p R_{n-1}(p) + \frac{1-p}{H_n} \sum_{k=1}^{n-1} \frac{1}{n-k} R_{k}(p)
\end{eqnarray}
Before proving the theorem we proceed with the following lemma.

\begin{lemma}\label{lem: bounds from previous recurrences}
Let $R_n=R_n(0), R_n(p)$ be the solutions to the recurrence relations \eqref{maineq1zz}, \eqref{maineq2zzipf},  respectively. Then for every $0\leq p <1$ and for every constants $c_1, c_2 >0$ with $c_1<4$, we have 
if 
$
\left( \forall n \geq 40, ~c_1 H_n^2 \leq R_n \leq c_2 H_n^2  \right)
$
then
$$
\left( \forall n \geq 40, ~\frac{4c_1/9}{1-(1-0.214c_1)p} H_n^2 \leq R_n(p) \leq \frac{c_2}{1-p} H_n^2 \right).
$$
assuming that the bounds for $R_n$ hold for $n=40$.\footnote{Constant $c_1$ is scaled by 4/9 only to satisfy a precondition in a subsequent theorem.}
\end{lemma}

\begin{proof} ({\bf Lemma~\ref{lem: bounds from previous recurrences}})
The proof is by induction on $n$, and the base case $n=40$ is straightforward. 

For the inductive step, suppose that the bounds for $R_n(p)$ are true for all integers up to $n-1$, and fix some $\Box \in \{\geq, \leq\}$ corresponding to the bounding constants $c \in \{c_1, c_2\}$ and $x \in \left\{\frac{c_1}{1-(1-0.214c_1)p},\frac{c_2}{1-p}\right\}$ respectively. 

The we have
\begin{align}
R_n(p) 
&
= 1+pR_{n-1}(p)+\frac{1-p}{H_n} \sum_{k=1}^{n-1} \frac{R_k(p)}{n-k} \tag{definition of $R_n$} \\
&
\Box~~ 1+p x H_{n-1}^2+\frac{(1-p)x}{H_n} \sum_{k=1}^{n-1} \frac{H_k^2}{n-k} \tag{Inductive Hypothesis} \\
&
\Box~~ 1+p x H_{n-1}^2+(1-p)x \left( H_n^2 - \frac1c \right) \tag{Preconditions} \\
&
\Box~~ 
x \left( pH_{n-1}^2 + (1-p) H_n^2 \right)
1 - \frac{(1-p)x}c \label{equa: conclusion to bound}
\end{align}

Now consider $\Box ="\geq"$, and observe that 
\begin{align}
R_n(p) 
&\stackrel{\eqref{equa: conclusion to bound}}{\geq}
x \left( pH_{n-1}^2 + (1-p) H_n^2 \right)
1 - \frac{(1-p)x}{c_1} \notag \\
& = x H_n^2 + xp(H_{n-1}^2 - H_n^2) + 1 - \frac{(1-p)x}{c_1} \notag \\
& = x H_n^2 + xp(H_{n-1} - H_n)(H_{n-1} + H_n) + 1 - \frac{(1-p)x}{c_1} \notag \\
& = x H_n^2 - xp \frac{H_{n-1} + H_n}{n} + 1 - \frac{(1-p)x}{c_1} \notag \\
& \geq x H_n^2 - 2xp \frac{H_n}{n} + 1 - \frac{(1-p)x}{c_1} \notag \\
& \geq x H_n^2 - 0.214 xp + 1 - \frac{(1-p)x}{c_1} \tag{Since $\frac{H_n}{n} < 0.106964$, for $n\geq 40$ } \\
& \geq x H_n^2 \tag{$x=\frac{4c_1/9}{1-(1-0.214c_1)p} \leq \frac{c_1}{1-(1-0.214c_1)p}$}
\end{align}

Finally, we consider $\Box ="\leq"$, and we have 
\begin{align}
R_n(p) 
&\stackrel{\eqref{equa: conclusion to bound}}{\leq}
x \left( pH_{n-1}^2 + (1-p) H_n^2 \right)
1 - \frac{(1-p)x}{c_2} \notag \\
& \leq x H_n^2 + 1 - \frac{(1-p)x}{c_2} \notag \\
& = x H_n^2 \tag{$x=\frac{c_2}{1-p}$}
\end{align}
This completes the proof of Lemma~\ref{lem: bounds from previous recurrences}.
\qed
\end{proof}

Now we proceed with the main proof of Theorem~\ref{thm2zipfr}.
Recall that by Lemma~\ref{costis2} we have $\frac{100}{383} H_n^2 \leq R_n \leq \frac57 H_n^2$ for all $n \geq 40$, where $R_n$ is the solution to the recurrence \eqref{maineq1zz}. But then, according to Lemma~\eqref{lem: bounds from previous recurrences}, it suffices to verify that for all $0 \leq p < 1$, both bounds below hold true
$$
~\frac{4c_1/9}{1-(1-0.214c_1)p} H_{40}^2 \leq R_{40}(p) \leq \frac{c_2}{1-p} H_{40}^2 
$$
where $c_1=100/383$ and $c_2=5/7$. In other words, it suffices to verify that 
$$
\frac{2.13}{1-0.945 p}
\leq R_{40}(p)
\leq 
\frac{13}{1-p}, ~~~\forall~ 0 \leq p < 1.
$$
Expression $R_{40}(p)$ is a polynomial on $p$ of degree 39, which can be computed explicitly from recurrence~\eqref{maineq2zzipf}. 
\begin{align*}
R_{40}(p)= &
3.70962710339202 \times 10^{-7} p^{39}
+3.0614726926339265\times 10^{-6} p^{38}
+0.0000139932 p^{37}\\
&+0.0000467865 p^{36}
+0.000127738 p^{35}+0.000301798 p^{34}+0.000639203 p^{33} \\
&+0.00124237 p^{32}+0.0022527 p^{31}+0.00385706 p^{30}
+0.00629362  p^{29} 
+0.00985709 p^{28}\\
&+0.0149033 p^{27}
+0.0218533 p^{26}
+0.0311969 p^{25}
+0.0434963 p^{24} 
+0.0593899 p^{23}\\
&+0.0795964 p^{22}
+0.104921 p^{21}
+0.136261 p^{20}
+0.174618 p^{19}
+0.221108 p^{18} 
+0.27698 p^{17} \\
&+0.343639 p^{16}
+0.422678 p^{15}
+0.51592 p^{14}
+0.625477 p^{13}
+0.753831 p^{12}
+0.903948 p^{11}\\
&+1.07944 p^{10} 
+1.28482 p^9
+1.52585 p^8+1.81016 p^7+2.14819 p^6+2.55498 p^5+3.05352 p^4 \\
&+3.68202 p^3+4.51248 p^2 
+5.7117 p+7.8824.
\end{align*}

\begin{figure}[!htb]
\begin{center}
\includegraphics[width=8cm]{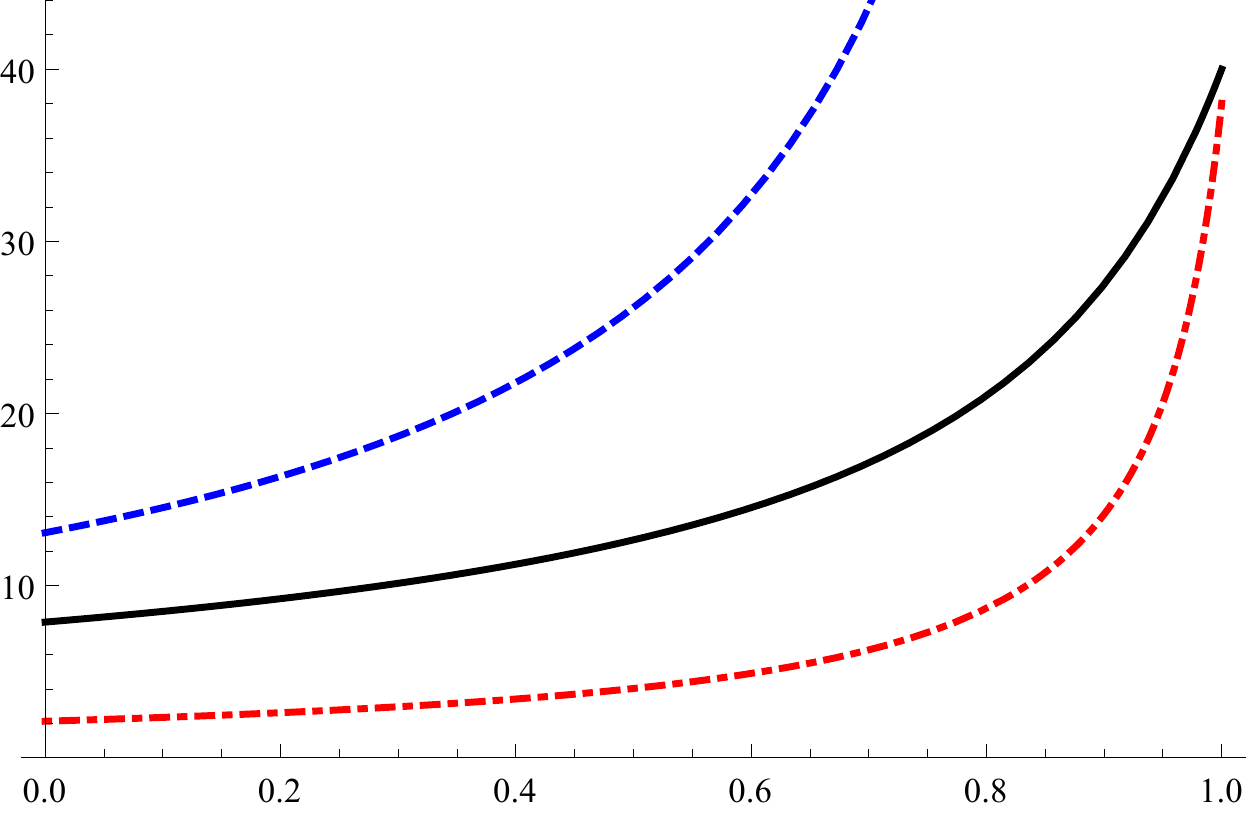}
\caption{The black solid line represents polynomial $R_{40}(p)$, the red dot-dashed line represents $\frac{4c_1/9}{1-(1-0.214c_1)p} H_{40}^2$, while the blue dotted line represents $\frac{c_2}{1-p} H_{40}^2$.}
\label{fig: sandwich R40}
\end{center}
\end{figure}
Then we can draw $R_{40}(p)$ to verify that it is indeed sandwiched between $\frac{2.13}{1-0.945 p}$ and $\frac{13}{1-p}$, for all $0 \leq p < 1$, as Figure~\ref{fig: sandwich R40} confirms 
Note that $\frac{13}{1-p}$ is unbounded as $p\rightarrow 1$, and hence its value exceeds $R_{40}(1)$ for $p$ large enough, here approximately for $p\geq 0.7$.
This completes the proof of Theorem~\ref{thm2zipfr}. \qed
\end{proof}

%
%

\section{The Occupancy of a Theater}
\label{theater:sec}

Given the previous results it is now easy to analyze the occupancy of a theater.
A typical theater consists of an array of rows separated by aisles.
This
naturally divides each row into sections which either have one entrance (e.g.,
when the row section ends with a wall) or two entrances. 
For example
in Figure \ref{lipari-fig} we see the Greek theatre on Lipari consisting of twelve
rows each divided into two one entrance sections and three two entrance sections.
In a sequential arrival model of
theatregoers, we assume that a theatergoer chooses a row and an entrance to the row by some arbitrary strategy. If she finds the row blocked at the entrance, then she moves on to the other entrance or another row. Then, the resulting occupancy of the theater will be equal to the sum of the number of occupied seats in each row of each
section. 
These values depend only on the length of the section. This
provides us with a method of estimating the total occupancy of the theatre. 


\begin{figure}[!h]
\begin{center}
\includegraphics[width=8cm]{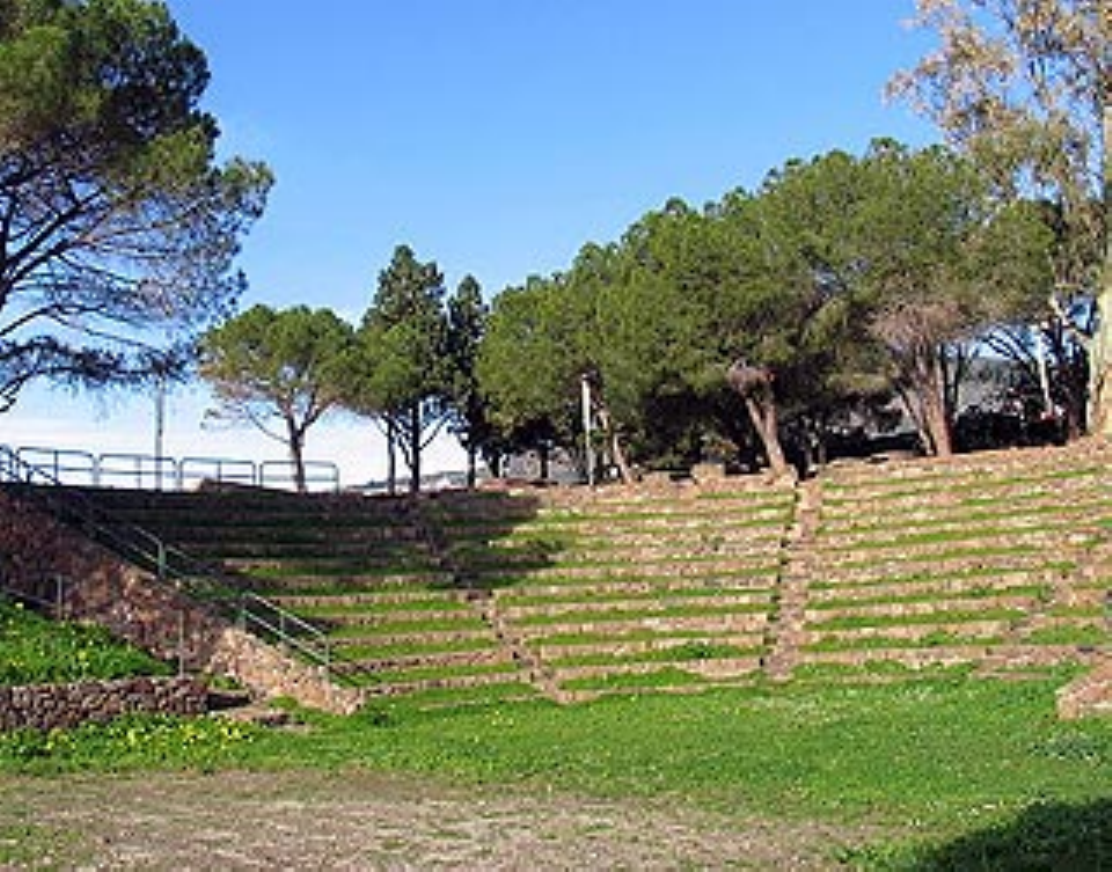}
\end{center}
\caption{The Greek theatre on Lipari Island.}
\label{lipari-fig}
\end{figure}

For example, for the Lipari theatre if each row section seats $n$ theatregoers then
we get the following:
\begin{corollary}
Consider a
theater having twelve rows with three aisles where each 
section contains $n$ seats. For firxed $0 <p<1$, the
expected number of occupied seats
assuming $p$-courteous
theatregoers is given by the expression
\begin{equation}
\label{pach3a}
-\frac{36}{1-p} + 96\frac{H_n - \ln (1-p)}{1-p},
\end{equation}
asymptotically in $n$.
\qed
\end{corollary}

\section{Conclusions and Open Problems}
\label{other:sec}

There are several interesting open problems worth investigating
for a variety of models reflecting alternative and/or changing 
behaviour of the
theatregoers, as well as their behaviour as a group.
Also problems arising from
the structure (or topology) of the theatre are interesting.
In this section we 
propose several open problems and directions for further research.
 

While we considered the uniform, geometric and Zipf distributions above,
a natural extension of the theatregoer model is to arbitrary distributions
with the probability that a theatregoer selects seat numbered
$k$ is $p_k$.
For example, theatregoers may prefer seats either not too close or too far from the stage.
These situations might introduce a bias 
that depends on the two dimensions of the position selected.
It would be interesting to compare the results obtained to
the actual observed occupancy distribution of a real open seating theatre
such as movie theatres in North America.

Another model results when the courtesy of a theatregoer
depends on the position selected, e.g., the further away from
an entrance the theatregoer is seated the less likely (s)he is
to get up.
Another interesting question arises when theatregoers
not only occupy seats for themselves but also 
need to reserve seats for their friends in a group.
Similarly, the courtesy of the theatregoers may now
depend on the number of people in a group, e.g.,
the more people in a group the less likely for all
theatregoers to get up to let somebody else go by. 
Another possibility is to
consider the courteous theatregoers problem in an
arbitrary graph $G= (V, E)$. Here, the seats
are vertices of the graph. Theatregoers
occupy vertices of the graph while new incoming theatregoers
occupy vacant vertices when available and may request 
sitting theatregoers to get up so as to allow them passage
to a free seat. Further, the set of nodes of the graph
is partitioned into a set of rows or paths of seats and a set
of ``entrances'' to the graph.
Note that in this more general case there could be
alternative paths to a seat. 
In general graphs, algorithmic questions arise such as 
give an algorithm that will maximize the 
percentage of occupied seats given that all theatregoers
are selfish.

%



\bibliographystyle{plain}
\bibliography{refs}

\appendix
\section{Proof of Lemma~\ref{lem: both bounds on summation} }

In what follows we fix some $\epsilon>0$.
Below we use that for every $f:\mathcal{R}_+\mapsto \mathcal{R}_+$ which is monotone, we have  
\begin{equation}\label{equa: integrals and sums}
\sum_{k=1}^{n-1} f(k) \leq \int_{1}^{n} f(t) \d t \leq \sum_{k=2}^{n} f(k).
\end{equation}

Then we observe that for every fixed $n$, expression $\frac{H_n-H_k}{n-k}$ is decreasing in $k$. This is because $\frac{-H_k}{n-k}$ is clearly decreasing, and the rate of change dominates that of the increasing expression $\frac{H_n}{n-k}$, since $n$ is fixed. That will be shortly combined with observation \eqref{equa: integrals and sums}. 

First we upper bound the sum above. 
\begin{align}
\sum_{k=1}^{n-1}\frac{H_n-H_k}{n-k}
&
=
\frac{H_n-H_{n-1}}{1}+ \sum_{k=1}^{n-2}\frac{H_n-H_k}{n-k} \notag \\
&=
\frac{1}{n}+ \sum_{k=1}^{n-2}\frac{H_n-H_k}{n-k} \notag \\
&\leq 
\frac{1}{n}+ \int_{1}^{n-1}\frac{H_n-H_k}{n-k} \d k  \tag{by \eqref{equa: integrals and sums}, since $\frac{H_n-H_k}{n-k}$ is monotone }\\
&= 
\frac{1}{n}+ \int_{1}^{n-1}\frac1{n-k}\sum_{t=k+1}^n \frac1t \d k  \notag \\
&\leq 
\frac{1}{n}+ \int_{1}^{n-1}\frac1{n-k}\int_{k+1}^{n+1} \frac1t \d t\d k \tag{by \eqref{equa: integrals and sums}, since $\frac{1}{t}$ is monotone } \\
&= 
\frac{1}{n}+ \int_{1}^{n-1}\frac1{n-k} \left( \ln (n+1) - \ln (k+1) \right) \d k \notag
\end{align}
By solving the last integral it is easily seen that the
last term is equal to

\begin{align}
&
\frac{1}{n}+ \ln(n+1)\ln \left(1+\frac{1}{n-1}\right) + \ln 2 \ln \left(1+\frac2{n-1}\right) 
- \Li_2\left(\frac2{n+1}\right) + \Li_2\left(1-\frac1{n}\right) \notag \\
&\leq 
2 \ln(n+1)\ln \left(1+\frac2{n-1}\right) 
- \Li_2\left(\frac2{n+1}\right) + \Li_2\left(1-\frac1{n}\right) \notag \\
&\leq 
\frac{4  \ln(n+1)}{n-1}
- \Li_2\left(\frac2{n+1}\right) + \Li_2\left(1-\frac1{n}\right) \label{equa: last upper bound}
\end{align}

Next we observe that both $\Li_2\left(\frac2{n+1}\right), \Li_2\left(1-\frac1{n}\right)$ are non negative, and in particular $\Li_2\left(\frac2{n+1}\right)$ is decreasing with 
$$
\lim_{n\rightarrow \infty}\Li_2\left(\frac{2}{n+1}\right) = 0,
$$
and $\Li_2\left(1-\frac1{n}\right)$ is increasing with 
$$\lim_{n\rightarrow \infty}\Li_2\left(1-\frac1n\right) = \frac{\pi^2}6.$$
Since also $\frac{4  \ln(n+1)}{n-1}$ is positive and tends to 0, we conclude that for big enough $n$ we have $\sum_{k=1}^{n-1}\frac{H_n-H_k}{n-k} \leq \pi^2/6 + \epsilon$. In particular, expression~\eqref{equa: last upper bound} is at most $1.86$ for all $n\geq 40$.

Next we lower bound the sum. 
\begin{align}
\sum_{k=1}^{n-1}\frac{H_n-H_k}{n-k}
&\geq 
\sum_{k=2}^{n-1}\frac{H_n-H_k}{n-k} \notag \\
&\geq 
\int_{1}^{n-1}\frac{H_n-H_k}{n-k} \d k  \tag{by \eqref{equa: integrals and sums}, since $\frac{H_n-H_k}{n-k}$ is monotone }\\
&= 
\int_{1}^{n-1}\frac1{n-k}\sum_{t=k+1}^n \frac1t \d k  \notag  \\
&\geq 
\int_{1}^{n-1}\frac1{n-k}\int_{k}^{n-1} \frac1t \d t\d k \tag{by \eqref{equa: integrals and sums}, since $\frac{1}{t}$ is monotone } \\
&= 
\int_{1}^{n-1}\frac1{n-k} \left( \ln (n-1) - \ln (k) \right) \d k \notag  \\
&= 
\ln(n-1)\ln \left(1-\frac{1}{n}\right) - \Li_2\left(\frac1{n}\right) + \Li_2\left(1-\frac1{n}\right) 
\label{equa: last lower bound}
\end{align}
As before, one can see that as $n$ tends to infinity, both expressions $\ln(n-1)\ln \left(1-\frac{1}{n}\right)$ and $\Li_2\left(\frac1{n}\right)$ tend to 0, while $\Li_2\left(1-\frac1{n}\right) $ tends to $\pi^2/6$. Hence, for big enough $n$ we have 
$\sum_{k=1}^{n-1}\frac{H_n-H_k}{n-k} \geq \pi^2/6 - \epsilon$. In particular, expression \eqref{equa: last lower bound} is at least 1.408 for all $n \geq 40$. 
\qed

\end{document}